\documentclass[a4paper]{llncs}
\usepackage{graphicx}
\usepackage{stuff,macros-jonas,tikz,hypersequents}

\title{Contraction-free proofs and finitary games for Linear Logic}

\institute{{CNRS, Universit\'e de Nice - Sophia Antipolis}
    \and {CEA - LIST} \and {CNRS, Universit\'e de Savoie}}

\author{{Andr\'e Hirschowitz}\inst{1} \and {Michel Hirschowitz}\inst{2}
  \and {Tom Hirschowitz}\inst{3}\thanks{MoDyFiable and Choco ANR
    projects}} 

\begin{document}
\maketitle

\begin{abstract} 
  In the standard sequent presentations of Girard's Linear
  Logic~\cite{ll} (LL), there are two "non-decreasing" rules, where
  the premises are not smaller than the conclusion, namely the cut and
  the contraction rules. It is a universal concern to eliminate the
  cut rule. We show that, using an admissible modification of the
  tensor rule, contractions can be eliminated, and that cuts can be
  simultaneously limited to a single initial occurrence.

  This view leads to a consistent, but incomplete game model for LL
  with exponentials, which is \emph{finitary}, in the sense that each
  play is finite. The game is based on a set of inference rules which
  does not enjoy cut elimination. Nevertheless, the cut rule is valid
  in the model.
\end{abstract}

\section{Overview}\label{sec:overview}

In an effort to strengthen the connection between
\begin{itemize}
\item categories of sequent calculus proofs modulo cut elimination and
\item categories of winning strategies,
\end{itemize}
we try to 
push forward the idea of \emph{neutrality}, introduced by Delande and
Miller~\cite{Miller08} (see also our~\cite {H3:mall:long}). The idea
there is to understand provability in MALL (Multiplicative Additive
Linear Logic) as a graphical game\footnote{Graphical game here means a
  game where positions are graphs. We do not give any formal
  definition.}, as follows.

Positions of the game are graphs with edges labelled by formulae. The
edges adjacent to each vertex form a sequent, the vertex itself
representing a player trying to prove that sequent.  Each move acts
upon exactly one edge; its restrictions to each end of the edge
correspond to inference rules in the sequent calculus.  Finally,
a formula $A$ is valid when, on the following graph, the left-hand
vertex (\emph{\Proponent{}}~$\Pl$) has a winning strategy against the
right-hand vertex (\emph{Opponent}~$\Op$):
\begin{center}
  \begin{hs}
    \player{l}{0,0} %
    \opponent{r}{2,0} %
    \arrow{l}{r}{$A$} %
  \end{hs}~.
\end{center}

In this paper, we extend this approach to exponential connectives,
for which we recall the four standard rules:
\begin{mathpar}
  \inferrule{\Gam, A}{\Gam, \wn A} 
  \and
  \inferrule{\Gam}{\Gam, \wn A}
  \and
  \inferrule{\Gam, \wn A, \wn A}{\Gam, \wn A}
  \and 
  \inferrule{\wn \Gam, A}{\wn \Gam, \oc A}
\end{mathpar}
(where $\wn \Gam$ denotes a list of $\wn$ formulae).  

Interpreting these rules as moves in a graphical game raises the issue
of infinite plays, for which deciding who wins is problematic.

In this paper, we explore the following way around this problem.  We
first replace the usual (right) tensor rule by an admissible variant
    \begin{mathpar}
      \inferrule{\Gam, \wn \The, A \\ \Del, \wn \The,
        B} {\Gam, \Del, \wn \The, A \otimes B}~\NewTens{},
    \end{mathpar}
which is still somehow decreasing, yielding an equivalent set of
inference rules which we call \LLT{}.  We then show that every sequent
provable in \LLT{} admits a \emph{bounded} proof, i.e., a
contraction-free one where cuts are reduced to at most one initial
occurrence.


Using this, one may devise a graphical game for \LLT{} without
contraction, and derive from it a model of \LLT{} provability, which is
consistent (but incomplete). We sketch this, and explain why the
result is not satisfactory: the model validates $A \impll \oc A$.

We then investigate one possible explanation for this deficiency: the
absence of $n$-ary dereliction. Restoring it leads to a set of
inference rules which we call \LLTN{}, where $\oc A$ has the more
standard decomposition $\bigwith_{n \in \nat} A^{\tens n}$. The notion
of a bounded proof makes sense in \LLTN{}, and we show that bounded
proofs yield a model of \LLT{}, albeit an incomplete one.  We show
that \LLTN{} does not satisfy cut elimination, but does satisfy
admissibility of cut, i.e., if $\Gam, A$ and $\nop{A}, \Del$ have
bounded proofs, then so does $\Gam, \Del$.

\LLTN{} yields a more satisfactory graphical game, in which winning
strategies form a model of \LLTN{}, again consistent.  Analogously to
previous work~\cite{H3:mall:long}, this model is not complete w.r.t.\
\LLTN{} (there is a winning strategy for $\bot \otimes \bot$ for
instance). However, we are confident that our notion of local
strategy, defined in~\cite{H3:mall:long} for the MALL fragment, will
extend smoothly to \LLTN{}, and yield a complete model. 

We thus have (by composition) a graphical game model of \LL{} (which,
again, will not be complete for \LL{}, because \LLTN{} itself is
not). We finally prove that admissibility of cut holds in this model.

We may sum up our models of \LL{} in a diagram %
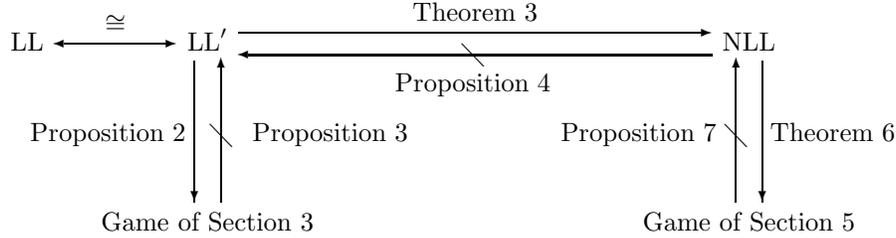
\begin{center}\begin{diagram}[tight,width=1.2cm,height=1.2cm,midshaft]
    \LL{} & \rCorresponds^{\iso} & \LLT{} & & &
    \pile{\rTo^{\mbox{Theorem~\ref{thm:bounded}}}  \\
      \lNegto_{\mbox{\Proposition{proposition:lltn:incomplete}}}} & & & \LLTN{} 
    & \ \ \ \ \ \ \ \ \ \ \ \ \ {} \\
    & & \dTo<{\mbox{\Proposition{proposition:soundness:naive}}}
    \uNegto>{\mbox{Proposition~\ref{proposition:unsatisfactory}}} & & & & & &
    \uNegto<{\mbox{Proposition~\ref{proposition:game:incomplete}}}
    \dTo>{\mbox{Theorem~\ref{thm:soundness}}} \\
    & & \mbox{Game of Section~\ref{sec:naive}} & & && && \mbox{Game of
      Section~\ref{sec:game}}
  \end{diagram}
\end{center}
where an edge between two notions of validity indicates that the
target is a model of the source, and a negated edge in the other
direction indicates that the model is not complete. We conjecture that
the vertical negated edges may be rectified by passing to local
strategies in the sense of~\cite{H3:mall:long}. However, the
horizontal negated arrow cannot be rectified.

\paragraph{Related work}
The idea of a game semantics for logic goes back at least to
Lorenzen~\cite{Lorenzen61}, and to Blass~\cite{blassgames} for linear
logic.  Blass' approach was extended by Abramsky et
al.~\cite{mllgames}, Hyland and Ong~\cite{compmllwomix}, Abramsky and
Melli\` es~\cite{concurgames,Mellies-AsynchronousGames4A}. Laurent
investigated the polarised case~\cite{phdlaurent}.  Among these
standard approaches, only some handle exponential connectives, and
among them none rule out infinite plays. Instead, they have to resort
to smart criteria to decide who wins these infinite plays.

Delande and Miller's~\cite{Miller08} game, our previous
work~\cite{H3:huet,H3:mall:long}, and Girard's
\emph{ludics}~\cite{locus} of course were a source of inspiration for
this paper. 

Kashima~\cite{Kashima} and Dyckhoff et
al.~\cite{Dyckhoff1992,DBLP:journals/igpl/DyckhoffN01} have already
observed certain forms of contraction elimination in other settings.

\paragraph{Organisation of the paper}
In \Section{sec:bounded:ll}, we prove that any sequent provable in
\LL{} has a bounded proof in \LLT{}. In \Section{sec:naive}, we build
a first graphical game from the rules of \LLT{} minus contraction. We
then prove that it yields a consistent model of \LLT{} provability,
which is however incomplete. In \Section{sec:rules}, we define our set
of inference rules with $n$-ary dereliction, called \LLTN{}, and prove
that it yields a model of \LLT{} provability, which is again
incomplete. We also show that \LLTN{} does not enjoy cut
elimination. We then move on in \Section{sec:game} to define our
graphical game for \LLTN{}, which we prove to yield a consistent model
of \LLTN{} (hence \LLTN{} is itself consistent), which is again
incomplete.  In \Section{sec:cut} we prove the cut rule to be
admissible in the model.

\paragraph{Acknowledgments} We warmly thank Olivier Laurent for his
constant ``Linear Logic Hotline'' and many useful examples and
counterexamples. We also thank Paul-Andr\'e Melli\`es for
encouragements, and Ren\'e David and Karim Nour for a very efficient
consulting session. Finally, we thank the anonymous referees for very
careful reading and helpful suggestions.

\section{Bounded proofs in LL}\label{sec:bounded:ll}
In this section, we prove our result on linear logic, namely that 
any provable sequent admits a bounded proof.

First, we recall LL formulae, defined by the grammar
$$\begin{array}{rc*{4}{c@{\ \ \alt\ \ }}c}
  A, B, C, \ldots & ::= & 
  \zero & \un & A \tens B & 
  A \plus B & \wn A \\
  & \alt & \top & \bot & A \parr B & A \with B & \oc A,
\end{array}$$
and decree that formulae on the first line are positive, the
others being negative. De Morgan duality, or linear negation, $\nop{A}$ 
is defined as usual
(sending a connective
to that vertically opposed to it). Observe that we do not handle 
propositional variables; the sequent calculus part extends easily to 
variables, but for games, this is not yet clear to us.
We use the usual priorities, e.g., $\oc A \impll B \tens C$ means
$(\oc A) \impll (B \tens C)$. Finally, sequents $\Gam, \Del, \ldots$ 
are lists of formulae, and we use $\Gam \vdash \Del$ as a notation
for $\nop{\Gam}, \Del$, when visually easier.

Now, consider the following variant of the tensor rule
    \begin{mathpar}
      \inferrule{\Gam, \wn \The, A \\ \Del, \wn \The,
        B} {\Gam, \Del, \wn \The, A \otimes B}~\NewTens{},
    \end{mathpar}
    which goes back at least to Andreoli~\cite{andreoli:phd}. We
    observe that it is derivable in LL.  Letting \LLT{} denote the set
    of inference rules obtained by replacing the usual tensor rule
    with~\NewTens{}, the following should be clear:
\begin{proposition}
  Provability in \LLT{} is equivalent to provability in \LL{}.
\end{proposition}

Thanks to rule \NewTens{}, we further have:
\begin{lemma}
  For any formula $A$, the formula $\dup{A} = \oc (\oc A \impll \oc A
  \tens \oc A)$ is provable with neither cuts nor contractions in
  \LLT{}, and furthermore for any $\Gam$, the rule
  \begin{mathpar}
    \inferrule{
      \Gam, \wn \nop{A}, \wn \nop{A}
    }{
      \Gam, \wn \nop{A}, \nop{\dup{A}}
    }~\textsc{Dup}
  \end{mathpar}
  is derivable without cuts nor contractions.
\end{lemma} (Here, $A \impll B$ denotes $\nop{A} \parr B$, as usual.)
\begin{proof}
Here is a proof of $\dup{A}$:
  \begin{mathpar}
\inferrule*{
  \inferrule*{
    \inferrule*{
      \wn A^\bot, \oc A \\
      \wn A^\bot, \oc A
    }{
      \wn A^\bot, \oc A \tens \oc A
    }
  }{
    \oc A \impll (\oc A \tens \oc A)
  }
}{ \oc (\oc A \impll (\oc A \tens \oc A)). }
\end{mathpar}
  Here is a derivation of \textsc{Dup}:
  \begin{mathpar}
    \inferrule*{ %
    \inferrule*{ %
      \inferrule*{\Gam, \wn A, \wn A %
      }{\Gam, \wn A \parr \wn A} \\ %
      \inferrule*{ }{ %
        \wn A, \oc A^\bot} %
    }{ \Gam, \wn A, \oc A^\bot \tens (\wn A \parr \wn A) } %
  }{ \Gam, \wn A, \wn (\oc A^\bot \tens (\wn A \parr \wn A)). }
\end{mathpar}
\qed\end{proof}
We call \emph{duplicators} the formulae of the shape $\dup{A}$.  Using
this, we construct \emph{bounded} proofs for provable sequents in
\LLT{}, in the following sense.

\begin{definition}
  A proof in \LLT{} is \emph{cc-free} when it is cut-free and
  contraction-free.  It is \emph{bounded} when it is either cc-free,
  or of the form $$ \inferrule{
    \inferrule{\pi_1}{A} \\
    \inferrule{\pi_2}{A^\bot, \Gam} \\
  }{\Gam}$$ with $\pi_1$ and $\pi_2$ cc-free.
\end{definition}

\begin{lemma}\label{lemma:cc}
  For any sequent $\Gam$ provable in \LLT{}, there is a list 
  $\dups$ of duplicators such that $\Gam, \nop{\dups}$ admits a
  cc-free proof.
\end{lemma}
(Here $\nop{\dups}$ denotes the list of duals of formulae in $\dups$.)
\begin{proof}
  By cut elimination in \LL{}, we may assume the given proof cut-free.
  The construction is then a mere induction on it, using the fact that
  duplicators are $\oc$ formulae. The non trivial case is contraction:
  from a given proof
  \begin{mathpar}
    \inferrule*{
      \inferrule*{
        \pi
      }{
        \Gam, \wn A, \wn A 
      }
      }{
        \Gam, \wn A
      }
  \end{mathpar}
  one obtains by induction hypothesis a proof $\pi'$ of $\Gam, \wn A,
  \wn A, \nop{\dups}$, for some list of duplicators $\dups$. One then
  derives (up to exchange):
  \begin{mathpar}
      \inferrule*[Right=Dup]{
        \inferrule*{\pi'}{\Gam, \wn A, \wn A, \nop{\dups}}
      }{
                \Gam, \wn A, \nop{\dups}, \nop{\dup{A}}.
      }
  \end{mathpar}
  (This does not use contraction or cut.)  For other cases, one just
  shows that a list of new hypotheses starting with $\wn$ do not
  hinder the derivation too much. This only adds weakenings and
  \textsc{Exchange} rules, but no cuts nor contractions.
\qed\end{proof}

Finally, we conclude:
\begin{theorem}\label{thm:bounded:ll}
   In \LLT{}, each provable sequent admits a bounded proof.
\end{theorem}
\begin{proof}
  Consider a provable sequent $\Gam$ in \LLT{}. By
  Lemma~\ref{lemma:cc}, $\Gam$ has a cc-free \LLT{} proof $\pi$ of
  $\Gam, \nop{\dups}$, for some $\dups = (\dup{1}, \ldots, \dup{n})$. So, the
  proof
  \begin{mathpar}
    \inferrule*{
      \inferrule*{
        \dup{1} \\
        \ldots \\
        \dup{n}
      }{ {}\tens{(\dups)}
      } \\
      \inferrule*{
        \inferrule*{\pi
        }{\Gam, \nop{\dups}
        } \\
      }{
        {}\parr{\nop{(\dups)}}, \Gam
      } 
    }{
      \Gam
    }
  \end{mathpar}
  of $\Gam$ is bounded.
\qed\end{proof}

\section{A naive game for \LLT{}}\label{sec:naive}
We have shown that we can dispense with the contraction rule of \LL{},
at the cost of a single initial cut and a modified tensor rule.

Starting from our previous game~\cite{H3:mall:long}, we now build an
extension of it to these rules.

\subsection{Positions }\label{subsec:positions}
As sketched above, the idea of positions in our game is as follows:
positions are graphs, whose edges are labelled by formulae. We want to
understand the neighborhood of a vertex, i.e., its adjacent edges, as
its sequent.  For instance, the vertex $\bullet$ in the graph
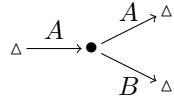
\begin{center}
  \begin{hs}
    \sequent{c}{0,0} %
    \player{l}{-1,0} %
    \player{ur}{1,.5} %
    \player{dr}{1,-.5} %
    \arrow{l}{c}{$A$} %
    \arrow{c}{ur}{$A$} %
    \arrowd{c}{dr}{$B$} %
  \end{hs}
\end{center}
sees the sequent $A, \nop{A}, B$, or equivalently $A \vdash A,
B$.

The use of graphs may be understood by seeing the cut rule
\begin{mathpar}
  \inferrule{\Gam \vdash A \\ A \vdash \Del}{\Gam \vdash \Del}
\end{mathpar}
as a move
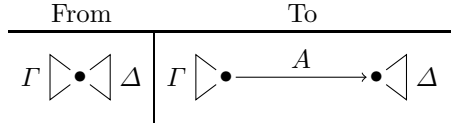
\begin{center}
  \begin{tabular}[t]{c|c}
    \multicolumn{1}{c}{From} & \multicolumn{1}{c}{To} \\ \hline
    \begin{hs}
      \sequent{c}{0,0} %
      \subtreerotwith{c}{90}{$\Gam$}{left} %
      \subtreerotwith{c}{-90}{$\Del$}{right} %
    \end{hs}
    &
    \begin{hs}
      \sequent{c}{0,0} %
      \sequent{r}{2,0} %
      \arrow{c}{r}{$A$} %
      \subtreerotwith{c}{90}{$\Gam$}{left} %
      \subtreerotwith{r}{-90}{$\Del$}{right} %
    \end{hs} 
  \end{tabular}
\end{center}
between graphs. Starting from the single vertex graph, after a finite
number of cuts, one may reach any connected and acyclic graph (in the
undirected sense). Hence our definition:
\begin{definition}
  A \emph{position} in our game is a \emph{punctured, directed,
    signed, normally labelled tree}.
\end{definition}
By \emph{punctured}, we mean that a vertex is selected; we say that
this vertex \emph{holds the token}. By \emph{directed} we mean that
edges are ordered pairs of vertices. By \emph{signed}, we mean that a
function from vertices to $\ens{O, P}$ is selected. Vertices labelled
$O$, called \emph{opponents}, are pictured by $\Op$, others, called
\emph{\proponents{}}, are pictured by $\Pl$. By \emph{labelled}, we mean
that a function from edges to formulae is selected, and by
\emph{normally labelled}, we mean that the selected formulae are
positive.  Finally by \emph{tree}, we mean that the underlying
undirected graph is a tree.  

We will also consider non-normally labelled trees, as a notation for
the position obtained by normalisation: for each edge labelled with a
negative formula, reverse altogether the edge and the formula.

There is an obvious notion of isomorphism, which turns positions into
a groupoid. Observe that positions may have symmetries
(automorphisms).

For each vertex, its adjacent edges determine (up to reordering) its
\emph{sequent}.


\subsection{Moves}\label{subsec:moves}

We now define our moves in Figure~\ref{fig:moves}. Moves are the edges
of a (huge) graph with positions as vertices, or equivalently a
reduction relation on positions.

Moves go left to right (in each row), and the vertex $v$ holding the
token, pictured as framed, \emph{plays} the move, or is \emph{active}.
The other vertex $v'$ is \emph{passive}.

The move is defined regardless of $v$ and $v'$ being \proponents{} or
an opponents. Thus, to avoid notational clutter, we define moves as a
relation between graphs, the polarity $\ens{O, P}$ being inherited
from the initial position. In the figure, $\bullet$ thus represents
vertices regardless of their polarity.

On the left, the \emph{broken} edge or formula is shown. It determines
two subtrees, which are pictured as triangles, and named to emphasise
the analogy with inference rules.  In the initial position of the
tensor move (third line), there are several triangles connected to the
left-hand vertex $v$: they denote subtrees which have $v$ as only
vertex in common.

\begin{figure*}[t]
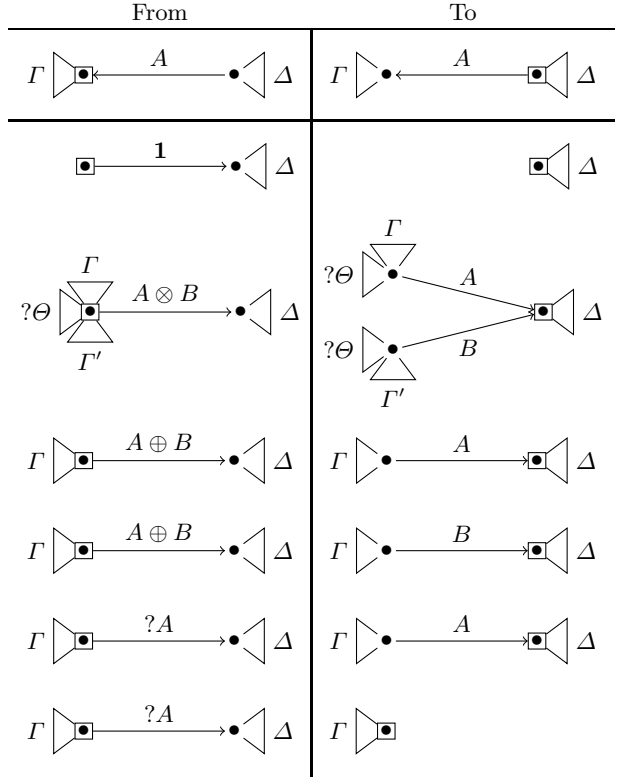

  \centering
  \begin{tabular}[t]{c|c}
    \multicolumn{1}{c}{From} & \multicolumn{1}{c}{To} \\ \hline
    \begin{hs}
      \asequent{c}{0,0} %
      \sequent{r}{2,0} %
      \arrow{r}{c}{$A$} %
      \subtreerotwith{c}{90}{$\Gam$}{left} %
      \subtreerotwith{r}{-90}{$\Del$}{right} %
    \end{hs}
    &
    \begin{hs}
      \sequent{c}{0,0} %
      \asequent{r}{2,0} %
      \arrow{r}{c}{$A$} %
      \subtreerotwith{c}{90}{$\Gam$}{left} %
      \subtreerotwith{r}{-90}{$\Del$}{right} %
    \end{hs} \\ \hline
    \begin{hs}
      \node (phant) at (-.8,0) {} ; %
      \asequent{c}{0,0} %
      \sequent{r}{2,0} %
      \arrow{c}{r}{$\un$} %
      \subtreerotwith{r}{-90}{$\Del$}{right} %
    \end{hs}
    &
    \begin{hs}
      \node (c) at (-.8,0) {}; %
      \asequent{r}{2,0} %
      \subtreerotwith{r}{-90}{$\Del$}{right} %
    \end{hs}
    \\
    \begin{hs}
      \asequent{c}{0,0} %
      \sequent{r}{2,0} %
      \subtreerotwith{c}{0}{$\Gam$}{above} %
      \subtreerotwith{c}{90}{$\wn \The$}{left} %
      \subtreerotwith{c}{180}{$\Gam'$}{below} %
      \subtreerotwith{r}{-90}{$\Del$}{right} %
      \arrow{c}{r}{$A \tens B$} %
    \end{hs}
    &
    \begin{hs}
      \sequent{u}{0,0.5} %
      \sequent{d}{0,-0.5} %
      \asequent{r}{2,0} %
      \subtreerotwith{u}{0}{$\Gam$}{above} %
      \subtreerotwith{u}{90}{$\wn \The$}{left} %
      \subtreerotwith{d}{90}{$\wn \The$}{left} %
      \subtreerotwith{d}{180}{$\Gam'$}{below} %
      \subtreerotwith{r}{-90}{$\Del$}{right} %
      \arrow{u}{r}{$A$} %
      \arrowd{d}{r}{$B$} %
    \end{hs}
    \\
    \begin{hs}
      \asequent{c}{0,0} %
      \sequent{r}{2,0} %
      \arrow{c}{r}{$A \plus B$} %
      \subtreerotwith{c}{90}{$\Gam$}{left} %
      \subtreerotwith{r}{-90}{$\Del$}{right} %
    \end{hs}
    &
    \begin{hs}
      \sequent{c}{0,0} %
      \asequent{r}{2,0} %
      \arrow{c}{r}{$A$} %
      \subtreerotwith{c}{90}{$\Gam$}{left} %
      \subtreerotwith{r}{-90}{$\Del$}{right} %
    \end{hs}
    \\
    \begin{hs}
      \asequent{c}{0,0} %
      \sequent{r}{2,0} %
      \arrow{c}{r}{$A \plus B$} %
      \subtreerotwith{c}{90}{$\Gam$}{left} %
      \subtreerotwith{r}{-90}{$\Del$}{right} %
    \end{hs}
    &
    \begin{hs}
      \sequent{c}{0,0} %
      \asequent{r}{2,0} %
      \arrow{c}{r}{$B$} %
      \subtreerotwith{c}{90}{$\Gam$}{left} %
      \subtreerotwith{r}{-90}{$\Del$}{right} %
    \end{hs}
    \\
    \begin{hs}
      \asequent{c}{0,0} %
      \sequent{r}{2,0} %
      \arrow{c}{r}{$\wn A$} %
      \subtreerotwith{c}{90}{$\Gam$}{left} %
      \subtreerotwith{r}{-90}{$\Del$}{right} %
    \end{hs}
    &
    \begin{hs}
      \sequent{c}{0,0} %
      \asequent{r}{2,0} %
      \arrow{c}{r}{$A$} %
      \subtreerotwith{c}{90}{$\Gam$}{left} %
      \subtreerotwith{r}{-90}{$\Del$}{right} %
    \end{hs}
    \\
    \begin{hs}
      \asequent{c}{0,0} %
      \sequent{r}{2,0} %
      \arrow{c}{r}{$\wn A$} %
      \subtreerotwith{c}{90}{$\Gam$}{left} %
      \subtreerotwith{r}{-90}{$\Del$}{right} %
    \end{hs}
    &
    \begin{hs}
      \node (r) at (2.8,0) {} ; %
      \asequent{c}{0,0} %
      \subtreerotwith{c}{90}{$\Gam$}{left} %
    \end{hs}
    %
  \end{tabular}
  \caption{Naive moves}
\label{fig:moves}
\end{figure*}

On the first row, we have \emph{negative} moves, where $A$ is required
to be positive. The active vertex $v$ passes the token along an input
edge. In some sense, $v$ has a negative formula, and asks the vertex
at the other end of this formula to break it. Thanks to acyclicity of
our positions, the token will come back, if ever, through the same
edge, breaking it at the same time. In the case of the negative
formula $\top$, the token will never come back, since there is no
right introduction rule for $\zero$. In this sense, the corresponding
negative move is winning.

Then we have positive moves. The active vertex may act on each of its
output formulae, except the $\zero$'s, and conditionally for the
$\un$'s.  The action is simple on formulae of the form $A \oplus B$:
it just changes the formula either into $A$ or into $B$. On $\wn A$,
the vertex may change it to $A$, or brutally erase $\wn A$, together
with its whole subtree. The action is pretty simple on the formula
$\un$: if (and only if) that edge is its sole adjacent edge, the
vertex may pass the token along that formula, and delete itself and
the edge from the tree.

It remains to explain the positive move in the case of a tensor
formula. This is the most complex move.  Here our active vertex, say
$v$, has an output edge $e$ labelled with $A \otimes B$, ending at the
neighbour vertex $v'$. When deleting $e$, we get a tree $\Del$ on the
$v'$ side, and another tree on the $v$ side, which $v$ has to split
into three subtrees $\Gam, \Gam'$, and $\wn \The$, sharing only their
root $v$. Furthermore, the subtree $\wn \The$ must exclusively consist
of outputs of the form $\wn C$. Then our vertex $v$ becomes two
vertices $v_1$ and $v_2$, where $v_1$ inherits $A$ (linked to $v'$),
$\Gam$ and a copy of $\wn \The$, while $v_2$ inherits $B$ (also linked
to $v'$), $\Gam'$ and a second copy of $\wn \The$. On the other side,
$v'$ keeps $\Del$ unchanged, and inherits both input formulae $A$ and
$B$, plus the token.

Our formalism allows to define cut moves and contraction moves (e.g.,
by changing $\wn A$ to $\wn A \parr \wn A$), but we leave them out of
the picture.

\subsection{Plays, strategies, validity}
We now define plays and (winning) strategies, and show that the latter yield 
a consistent model of \LLT{} provability (which is incomplete).

Having defined moves, plays on a position $U$ follow: they are
(directed, possibly countable) paths from $U$ in the graph of
positions and moves, or equivalently reduction sequences.  They happen
to be finite:
\begin{lemma}\label{lemma:finite:naive}
  In the above graph of positions and moves, there is no infinite path
  (i.e., in our game, all plays are finite).
\end{lemma}
\begin{proof}
  This is a consequence of our later Theorem~\ref{thm:finite}.
\qed\end{proof}

For a position $U$, we then define a \emph{strategy} to be a \emph
{welcoming}, prefix-closed set of plays from $U$, containing the empty
play, and stable under isomorphism. The latter means that a strategy
must contain two isomorphic plays
    \begin{diagram}
      U & \rTo & U_1 & \rTo & U_2 & \rTo & \ldots \\
      & \rdTo & \dTo~{\iso} & & \dTo~{\iso} & & \ldots \\
      & & V_1 & \rTo & V_2 & \rTo & \ldots
    \end{diagram}
at the same time.
By \emph{welcoming}, we mean that the
strategy accepts any move from opponents. In a strategy $\strat$, a
play is \emph{maximal} when $\strat$ contains no extension of it.
\begin{definition}
  A strategy is \emph{winning} if, in all its maximal positions, the
  token is held by an opponent.
\end{definition}

We then define validity. This is a bit more delicate, because the game
does not feature contraction, so for proofs to yield winning
strategies, we account for the needed initial cut:
\begin{definition}\label{def:valid}
  A formula $A$ is \GLLT{}-\emph{valid} (G stands for ``Game'') iff
  there exists a tree $U$ such that for all trees $V$, the vertices of
  $U$ have a winning strategy against those of $V$ in the graph
  \begin{center}
      \begin{hs}
    \player{l}{0,0} %
    \opponent{r}{2,0} %
    \subtreerotwith{l}{90}{$U$}{left} %
    \subtreerotwith{r}{-90}{$V$}{right} %
    \arrow{l}{r}{$A$} %
  \end{hs}
  \end{center} wherever the
  token is placed initially.
\end{definition}

Using the same methods as in Section~\ref{sec:logic} below, one may show
that \GLLT{}-validity yields a consistent model of \LLT{}:
\begin{theorem}
  \GLLT{}-validity is consistent, i.e., two dual formulae are not both
  \GLLT{}-valid.
\end{theorem}

\begin{proposition}\label{proposition:soundness:naive}
  Any formula provable in \LLT{} is \GLLT{}-valid.
\end{proposition}
The proofs are as in Section~\ref{sec:logic}, so we do not repeat them
here.

We consider this game unsatisfactory, because its
exponentials are terribly non standard:
\begin{proposition}\label{proposition:unsatisfactory}
  The formula $A \impll \oc A$ is \GLLT{}-valid.
\end{proposition}
We first observe:
\begin{lemma}
  Axioms $A \vdash A$ are \GLLT{}-valid.
\end{lemma}
Then:
\begin{proof}[of \Proposition{proposition:unsatisfactory}]
  One reaches a sequent $A \vdash \oc A$, where, roughly, one may pass
  the token along $\oc A$, and then either be erased (and hence win),
  or become $A \vdash A$ (and win again). 
\qed\end{proof}

The remainder of the paper patches this deficiency by allowing $n$-ary
dereliction. We start by considering a new set of inference rules.

\section{\LLTN{}}\label{sec:rules}

Let us start from the rules of $\LLT{}$. We want to patch them to add
$n$-ary dereliction. This means that instead of contraction,
weakening, and dereliction, we use the rules:
  \begin{mathpar}
    \inferrule{\Gam, \bot}{\Gam, \wn A} \and 
    \inferrule{\Gam, A}{\Gam, \wn A} \and \ldots \and
    \inferrule{\Gam, {\parr^n} A}{\Gam, \wn A}
    \and \ldots
  \end{mathpar}
  with $\parr^ 0 A= \bot$ and $\parr^{n +1} A= \parr^{n} A \parr A$.

To preserve the symmetry between dual connectives, we use for $\oc$ the rule 
\begin{mathpar}
  \inferrule{\Del, \un \\ \Del, A \\ \ldots \\ \Del, {\tens^{n}} A \\
    \ldots}{\Del, \oc A}
\end{mathpar}
which has infinitely many premises.

The idea behind these rules is to encode $\oc A$ as an infinitary
additive conjunction $\bigwith_{n \in \nat} A^{\tens n}$, which is an
old idea, at least as old as Lafont's work on categorical models of
\LL{}~\cite{DBLP:journals/tcs/Lafont88,DBLP:journals/tcs/Lafont88a}.

This yields a set of inference rules, which we sum up
in~\Figure{fig:sequent} and call $\LLTN{}$.  Since we have infinitary
rules, let us precisely define proofs. 
\begin{definition}
  A proof in \LLTN{} is a tree whose branches are all finite, and
  whose nodes are labelled with inference rules as usual.
\end{definition}
However, we are mainly interested in \emph{bounded} proofs, in the
following sense:
\begin{definition}\label{def:bounded}
  A proof in \LLTN{} is \emph{bounded} when it is either cut-free, or
  of the form
\begin{equation}
  \inferrule{
    \inferrule{\pi_1}{A} \\
    \inferrule{\pi_2}{A^\bot, \Gam} \\
  }{\Gam}
\label{eq:bounded}
\end{equation}
 with $\pi_1$ and $\pi_2$ cut-free.
\end{definition}

\begin{figure*}[!th]
  \centering \abovedisplayskip=0pt \belowdisplayskip=0pt
  \begin{tabular}{|p{.4\linewidth}|p{.4\linewidth}|}
    \hline
    \hfil Positive rules \hfil & \hfil Negative rules \hfil \\ \hline%
    \begin{mathpar}%
      \inferrule{ }{\un} %
      \end{mathpar}&
      \begin{mathpar}
        \inferrule{\Gam}{\Gam, \bot} 
      \end{mathpar}
      \\  \hline
      &
      \begin{mathpar}
        \inferrule{ }{\Gam, \TT} 
      \end{mathpar} \\ \hline
      \begin{mathpar}
        \inferrule{
          \Gam, \wn \Theta, A \\
          \Del, \wn \Theta, B }{ \Gam, \Del, \wn \Theta, A \tens B
        } 
      \end{mathpar}
      &
      \begin{mathpar}
        \inferrule{\Gam, A, B}{\Gam, A \parr B} 
      \end{mathpar} \\ \hline
      \begin{mathpar}
        \inferrule{\Gam, A}{\Gam, A \plus B}  \and
        \inferrule{\Gam, B}{\Gam, A \plus B}  
      \end{mathpar}
      & 
      \begin{mathpar}
        \inferrule{\Gam, A \\ \Gam, B}{\Gam, A \with B} 
      \end{mathpar} \\ \hline
      \begin{mathpar}
        \inferrule{\Gam, {\parr^n} A}{\Gam, \wn A} 
      \end{mathpar}
      & 
      \begin{mathpar}
  \inferrule{\Gam, \un \\ \ldots \\ \Gam, {\tens^{n}} A \\
    \ldots}{\Gam, \oc A} 
      \end{mathpar}
      \\ \hline
      \multicolumn{2}{|p{.8\linewidth}|}{  
  \begin{mathpar}
    \inferrule{\Gam, B, A, \Del}{\Gam, A, B, \Del} 
    \and
    \and
    \inferrule{\Gam, A \\ \Del, A^\bot}{\Gam, \Del} 
    \end{mathpar}    
  }   \\ \hline
\end{tabular}
\caption{Inference rules for \LLTN{}}
\label{fig:sequent}
\end{figure*}

We now show that:
\begin{itemize}\item 
  \LLTN{} is a model of \LL{} provability;
\item This model is incomplete, i.e., the converse does not hold;
\item $\LLTN{}$ does not enjoy cut elimination.
\end{itemize}

\begin{lemma}
  Axioms $A \vdash A$ have bounded proofs in \LLTN{}.
\end{lemma}

\begin{theorem}\label{thm:bounded}
  Each provable sequent in \LLT{} admits a bounded proof in \LLTN{}.
\end{theorem}

\begin{proof}
  By induction on a given bounded proof in \LLT{}, we construct a
  bounded proof (of the same sequent) in \LLTN{}. All cases but
  promotion are obvious by induction hypothesis (because the rules of
  \LLT{} are almost the same as those of \LLTN{}). For promotion,
  assume given a proof of $\wn \Gam, A$. We need to find an \LLTN{}
  proof of each $\wn \Gam, \tens^n A$: for each $n$, we split the
  $n$-fold tensor, and reduce to proving $\wn \Gam, A$ or $\wn \Gam,
  \un$, which are both have bounded proofs in \LLTN{} (the former by
  induction hypothesis, the latter straightforwardly).  Since we only
  introduce tensor and $\un$ rules, this results in a bounded proof.
\qed\end{proof}

We have thus shown that bounded proofs in \LLTN{} are a model of
\LLT{} provability.  However, the converse does not hold, thanks to an
example of Olivier Laurent:
\begin{proposition}\label{proposition:lltn:incomplete}
  \LLT{} is not a model of \LLTN{}.
\end{proposition}
\begin{proof}
  The sequent $\bot \plus (X^\bot \parr \wn X^\bot), \oc X$ has a
  bounded proof in \LLTN{}, as follows:
  \begin{mathpar}
    \inferrule*{
      \inferrule*{\inferrule*{\dots}{\bot, \un}}{\bot \plus (X^\bot \parr \wn X^\bot), \un} \\
      \inferrule*{\inferrule*{X^\bot, X}{X^\bot \parr \wn X^\bot, X}}{\bot \plus (X^\bot \parr \wn X^\bot), X} \\
      \ldots }{ \bot \plus (X^\bot \parr \wn X^\bot), \oc X.  }
  \end{mathpar}
  However, it is not provable in \LL{}.  There could be a problem
  since it uses a variable, but taking $X = \bot \tens \bot$ makes
  everything work the same.
\qed\end{proof}
Finally:
\begin{proposition}
  \LLTN{} does not enjoy cut elimination.
\end{proposition}
\begin{proof}
  Consider the sequent $X \with ((X \with Y) \parr Y) \vdash \wn (X
  \with Y)$.  It admits the following \LLT{} proof:
  \begin{mathpar}
    \inferrule*
    {
      \inferrule*{
      \inferrule*{
        \inferrule*{
      \inferrule*{
        X \vdash X
      }{
        X \vdash \wn (X \with Y), X
      }
    }{
      X \with ((X \with Y) \parr Y) \vdash \wn (X \with Y), X
      } \\
      \inferrule*{
        \inferrule*{
        Y \vdash Y
        \\
      \inferrule*{
        X \with Y \vdash X \with Y
      }{
        X \with Y \vdash \wn (X \with Y)
      }
      }{
        ((X \with Y) \parr Y) \vdash \wn (X \with Y), Y
      }
      }{
        X \with ((X \with Y) \parr Y) \vdash \wn (X \with Y), Y
      }
      }{
        X \with ((X \with Y) \parr Y) \vdash \wn (X \with Y), (X \with Y)
      }
      }{
        X \with ((X \with Y) \parr Y) \vdash \wn (X \with Y), \wn (X \with Y)
      }
    }{
      X \with ((X \with Y) \parr Y) \vdash \wn (X \with Y),
    }
  \end{mathpar}
  hence has a bounded proof in \LLTN{} by Theorem~\ref{thm:bounded}.
  However, careful inspection shows that it has no cut-free proof.
  Again, we do this using variables, which are not in the syntax.  It
  all works the same by taking $X$ and $Y$ such that none entails the
  other, e.g., $X = \bot \tens \TT$ and $Y = 1 \parr \zero$.
\qed\end{proof}

Nevertheless, we have:
\begin{proposition}
  \LLTN{} features admissibility of cut, i.e., if $\Gam, A$ and
  $\nop{A}, \Del$ both have bounded proofs, then so does $\Gam, \Del$.
\end{proposition}
\begin{proof} 
  Assuming first that the given bounded proofs have the shape:
  \begin{mathpar}
    \inferrule*{
      \inferrule{\pi_1}{\Gam, A, B} \\
      \inferrule{\pi_2}{\nop{B}}
      }{
        \Gam, A
      } \and
    \inferrule*{
      \inferrule{\pi_3}{\Del, \nop{A}, C} \\
      \inferrule{\pi_4}{\nop{C}}
      }{
        \Del, \nop{A}
      }  \end{mathpar}
    with the $\pi_i$'s cut-free.
    We then form the bounded proof:
    \begin{mathpar}
      \inferrule*{
        \inferrule*{
          \inferrule*{
            \inferrule*{\pi_1}{\Gam, A, B} \\
            \inferrule*{\pi_3}{\Del, \nop{A}, C}
          }{\Gam, \Del, (A \tens \nop{A}), B, C} 
        }{
          \Gam, \Del, (A \tens \nop{A}) \parr B \parr C
        } \\
        \inferrule*{
          \inferrule*{\inferrule*{\pi_2}{\nop{B}} \\
            \inferrule*{\pi_4}{\nop{C}}
          }{\nop{B} \tens \nop{C}} \\
          \nop{A}, A
        }{
          (\nop{A} \parr A) \tens \nop{B} \tens \nop{C}
        }
      }{ 
        \Gam, \Del.
      }
    \end{mathpar}
    Cases where not both proofs start with a cut behave essentially
    the same.
\qed\end{proof}

We now proceed to define a graphical game interpretation of \LLTN{}.

\section{A graphical game for \LLTN{}}\label{sec:game}
The positions of our game are exactly as in Section~\ref{sec:naive}.
Moves are as in Figure~\ref{fig:moves}, with both exponential moves
replaced by the single move scheme:
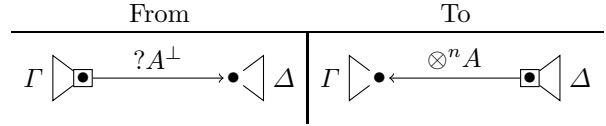
\begin{center}
  \begin{tabular}[t]{c|c}
    \multicolumn{1}{c}{From} & \multicolumn{1}{c}{To} \\ \hline
    \begin{hs}
      \asequent{c}{0,0} %
      \sequent{r}{2,0} %
      \arrow{c}{r}{$\wn A^\bot$} %
      \subtreerotwith{c}{90}{$\Gam$}{left} %
      \subtreerotwith{r}{-90}{$\Del$}{right} %
    \end{hs}
    &
    \begin{hs}
      \sequent{c}{0,0} %
      \asequent{r}{2,0} %
      \arrow{r}{c}{${\tens^{n}} A$} %
      \subtreerotwith{c}{90}{$\Gam$}{left} %
      \subtreerotwith{r}{-90}{$\Del$}{right} %
    \end{hs}
  \end{tabular}
\end{center}
(using the conventions of Figure~\ref{fig:moves}).

Plays are as in Section~\ref{sec:naive}: (directed, possibly
countable) paths in the graph of positions and moves.  
As in \Section{sec:naive}, we have:
\begin{theorem}\label{thm:finite} All plays are finite.
\end{theorem}
We prove this by adapting to our graphical setting the syntactic
techniques of David and Nour~\cite{David01,David03}.

We will need an auxiliary move, which we call the \emph{exotic} move:
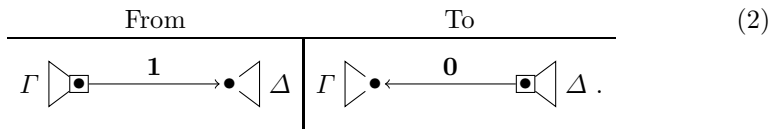
\begin{equation}
  \mbox{\begin{tabular}[t]{c|c}
    \multicolumn{1}{c}{From} & \multicolumn{1}{c}{To} \\ \hline
    \begin{hs}
      \asequent{c}{0,0} %
      \sequent{r}{2,0} %
      \arrow{c}{r}{$\un$} %
      \subtreerotwith{c}{90}{$\Gam$}{left} %
      \subtreerotwith{r}{-90}{$\Del$}{right} %
    \end{hs}
    &
    \begin{hs}
      \sequent{c}{0,0} %
      \asequent{r}{2,0} %
      \arrow{r}{c}{$\zero$} %
      \subtreerotwith{c}{90}{$\Gam$}{left} %
      \subtreerotwith{r}{-90}{$\Del\ .$}{right} %
    \end{hs} 
  \end{tabular}}\label{eq:new:move}
\end{equation}

Call a position $U$ \emph{strongly normalizable} (sn for short) if,
even if we add the above exotic move, all plays from $U$ are finite.  
Why do we need a new move? We want to be sure that if a
position is sn, so are all its subpositions, in the sense of
subgraphs. And the problem with this is that, although most moves on
a subposition come from a move on the full position, this is not the
case for a $\un$-move. Indeed, consider the bare $\un$-move
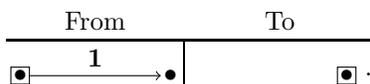
\begin{center}
  \begin{tabular}[t]{c|c}
    \multicolumn{1}{c}{From} & \multicolumn{1}{c}{To} \\ \hline
    \begin{hs}
      \asequent{c}{0,0} %
      \sequent{r}{2,0} %
      \arrow{c}{r}{$\un$} %
    \end{hs}
    &
    \begin{hs}
      \phant{c}{0,0} %
      \asequent{r}{2,0} %
    \end{hs}\ .
  \end{tabular}
\end{center}
If we embed the initial position into the bigger
\begin{center}
    \begin{hs}
      \sequent{l}{-2,0} %
      \asequent{c}{0,0} %
      \sequent{r}{2,0} %
      \arrow{c}{r}{$\un$} %
      \arrow{c}{l}{$\zero$} %
    \end{hs}  
\end{center}
we obtain a position without any possible move. Our exotic move allows
the token to be passed to the right. The intuition is that it
transforms $\un$-moves into ``one-way'' moves, merely passing the
token along $\un$ without any hope for return. Which has roughly the
same effect as deleting the edge.

Thanks to the exotic move, we have:
\begin{lemma}
  Any subposition of an sn position is again sn.
\end{lemma}

Now, consider a tuple $U = (U, v, a_1, \ldots, a_n)$ where $U$ is a
position with marked vertices $a_1, \ldots, a_n$, and $v$ is the
vertex owning the token. Consider also a tuple $V = ((V_1, b_1),
\ldots, (V_n, b_n))$ of positions ($b_i$ is the node having the
token), and a formula $A$.  From these data, we build a new position
$W = W(U, V, A)$ by taking the union of $U$ with the $V_i$'s, and
adding $n$ new edges from $a_i$ to $b_i$ for each $i$, labelling these
with $A$ (which could be negative), and putting the token at $v$.

Using this notation, the theorem follows (by an easy induction on the
number of vertices) from the following lemma:
\begin{lemma}
If $U$ and the $V_i$'s are sn, then so is $W$.
\end{lemma}

\begin{proof}
  We proceed by lexicographic induction on the triple $(A, U, n)$,
  using the following orderings:

  \begin{description}\item[Formulae:] We use the ordering determined
    by the following four rules:
    \begin{itemize}
    \item A formula is greater than each of its subformulae and their
      duals;

    \item $\un$ is greater than $\top$;

    \item A formula $\wn C$ is greater than $ \parr^n C$ for each $n$
      (and similarly for $\oc$ and $\tens$);

    \item A negative formula is greater than its linear negation.

    \end{itemize}

  \item[Positions:] 
    We order sn positions by putting $U \geq V$ iff there is a
    path in the graph of moves, or equivalently a reduction sequence,
    from $U$ to $V$. Since we limit this to sn positions, it is
    well-founded.

  \item[Numbers:] 
    Finally, we use the standard ordering for natural numbers.
  \end{description}
  In order to prove that there is no infinite play starting at $W$, it
  is enough to show that all positions reached from $W$ after one move
  are sn. We proceed by a case analysis on this first move.

  If it is a move inside $U$: $U \to U'$, we apply the induction
  hypothesis to the new position $W_0$. We have to explain how it has
  the required form. Indeed, we have $W_0 = W(U', V', A)$, where:
  \begin{itemize}\item 
    the marked points in $U'$ remain the same, except in the case of a
    $\un$ move, where some $a_i$'s may be deleted, and in the case of a
    $\tens$ move where some $a_i$'s may be duplicated;

  \item $V'$ is like $V$, except in case the move is a $\un$ move,
    where possibly some $V_i$'s have to be deleted, and in case of a
    $\tens$ move, where some $V_i$'s have to be duplicated.
  \end{itemize}
  Hence in the present case, $n$ may increase, but $A$ is constant, and
  $U$ strictly decreases.

  Suppose now $v = a_{j}$ and our move $W \to W_0$ is on the 
  $j$-th $A$ edge between $a_{j}$ and $b_{j}$.

  We treat separately the most delicate case where $A = B \tens C$. 
  The move then looks like:
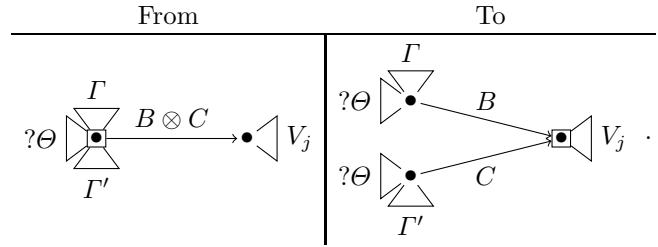
\begin{center}
  \begin{tabular}[t]{c|c}
    \multicolumn{1}{c}{From} & \multicolumn{1}{c}{To} \\ \hline
    \begin{hs}
      \asequent{c}{0,0} %
      \sequent{r}{2,0} %
      \subtreerotwith{c}{0}{$\Gam$}{above} %
      \subtreerotwith{c}{90}{$\wn \The$}{left} %
      \subtreerotwith{c}{180}{$\Gam'$}{below} %
      \subtreerotwith{r}{-90}{$V_j$}{right} %
      \arrow{c}{r}{$B \tens C$} %
    \end{hs}
    &
    \begin{hs}
      \sequent{u}{0,0.5} %
      \sequent{d}{0,-0.5} %
      \asequent{r}{2,0} %
      \subtreerotwith{u}{0}{$\Gam$}{above} %
      \subtreerotwith{u}{90}{$\wn \The$}{left} %
      \subtreerotwith{d}{90}{$\wn \The$}{left} %
      \subtreerotwith{d}{180}{$\Gam'$}{below} %
      \subtreerotwith{r}{-90}{$V_j$}{right} %
      \arrow{u}{r}{$B$} %
      \arrowd{d}{r}{$C$} %
    \end{hs}\ .
  \end{tabular}
\end{center}
But both positions
\begin{center}
  $U' = {}$\begin{hs}
      \sequent{c}{0,0} %
      \subtreerotwith{c}{0}{$\Gam$}{above} %
      \subtreerotwith{c}{90}{$\wn \The$}{left} %
  \end{hs}
\hfil and \hfil
  $U'' = {}$\begin{hs}
      \sequent{c}{0,0} %
      \subtreerotwith{c}{90}{$\wn \The$}{left} %
      \subtreerotwith{c}{180}{$\Gam'$}{below} %
  \end{hs}
\end{center}
are sn by induction hypothesis: they both have the shape $W (X, Y, A)$
with $X \leq U$ and $Y$ strictly shorter than $V$, because $V_j$ has
been removed. But $V_j$ itself is sn by hypothesis, so by induction
hypothesis again, the position
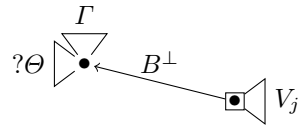
\begin{center}
  \begin{hs}
      \sequent{u}{0,0.5} %
      \asequent{r}{2,0} %
      \subtreerotwith{u}{0}{$\Gam$}{above} %
      \subtreerotwith{u}{90}{$\wn \The$}{left} %
      \subtreerotwith{r}{-90}{$V_j$}{right} %
      \arrow{r}{u}{$\nop{B}$} %
  \end{hs}
\end{center}
is again sn, since it is obtained as
$W' = W (V_j, U', \nop{B})$. Finally, by induction hypothesis,
$W_0 = W (W', U'', \nop{C})$ is also sn.

  For all other cases, we apply the inductive hypothesis twice: one for
  the new $V$, with smaller $n$, and one for the new $U$, with a
  smaller $A$.%
\qed\end{proof}

By the way, Lemma~\ref{lemma:finite:naive} follows by remarking that any
play in the game of Section~\ref{sec:naive} is simulated by a play in
the new game (with the exotic move~\eqref{eq:new:move}).

\section{Soundness}\label{sec:logic}

So we have a game for \LLTN{}.  Let us now show that it defines a
model of \LLTN{} (and hence of \LL{}).  The definition of (winning)
strategies is exactly as in Section~\ref{sec:naive},
thanks to Theorem~\ref{thm:finite}. 
As in \Section{sec:naive}, we put:

\begin{definition}
  A formula $A$ is \emph{\GLLTN{}-valid} (again, G stands for
  ``Game'') iff there exists a tree $U$ such that for all trees $V$,
  the vertices of $U$ have a winning strategy against those of $V$ in
  the graph
  \begin{center}
    \begin{hs}
      \player{l}{0,0} %
      \opponent{r}{2,0} %
      \subtreerotwith{l}{90}{$U$}{left} %
      \subtreerotwith{r}{-90}{$V$}{right} %
      \arrow{l}{r}{$A$} %
    \end{hs}
  \end{center} wherever the token is placed initially.
\end{definition}

We immediately have consistency:
\begin{theorem}\label{thm:consistency}
  A formula and its dual are not both \GLLTN{}-valid.
\end{theorem}
\begin{proof}
  Given two trees $U$ and $V$ as above, and strategies $S$ and $S'$
  for them, they both contain the empty play so $S \cap S'$ is
  non-empty. By Theorem~\ref{thm:finite}, we have a maximal play in $S
  \cap S'$. In its final position, the token is held by some vertex
  $v$. If $v$ is a \proponent{} then $S$ is not winning, otherwise $S'$ is
  not winning.
\qed\end{proof}
As a consequence, since \LLTN{} has a consistent model, we have:
\begin{corollary}
  \LLTN{} is consistent.
\end{corollary}
We now prove:
\begin{theorem}\label{thm:soundness}
  \GLLTN{}-validity is a model of \LLTN{} (bounded) provability.
\end{theorem}
We start with:
\begin{lemma}\label{lemma:total}
  A strategy is winning iff it is \emph{total}, i.e., for each reached
  position where a \proponent{} holds the token, it has an extension by a
  move.
\end{lemma}
We then observe:
\begin{lemma}\label{lemma:winning}
  For any position $U$, equipped with, for each \proponent{} vertex $v
  \in U$, a cut-free \LLTN{} proof $\pi_v$ of its sequent, there exists
  a winning strategy on $U$.
\end{lemma}

\begin{proof}
  Before constructing such a strategy, let us observe that the
  following rules are admissible in \LLTN{}, and moreover for any
  cut-free proof of their premise, there is a cut-free proof of their
  conclusion:
  \begin{mathpar}
    \inferrule{\Gam, \bot}{\Gam} \and
    \inferrule{\Gam, A \parr B}{\Gam, A, B} \and
    \inferrule{\Gam, A \with B}{\Gam, A} \and
    \inferrule{\Gam, A \with B}{\Gam, B} \and
    \inferrule{\Gam, \oc A}{\Gam, {\tens^{n}} A}~\cdot
  \end{mathpar}
  We say that going from premise to conclusion in one of these rules is
  an \emph{anodyne} modification.

  Assume now given a position $U$, and a cut-free proof $\pi_v$ for
  each \proponent{} $v \in U$, and let us construct a winning strategy
  on $U$. We describe this by defining a new graph $\G$, which embeds
  into the graph of positions and moves. It has as vertices all
  positions $U$ equipped with a cut-free proof $\pi_v$ for each
  \proponent{} $v \in U$. We now define edges from any such $U$.

  First assume the token is held by opponents in $U$. After each move
  $U \to V$, if the token is still held by opponents in $V$, then
  \proponent{} sequents remain unchanged, and we add an edge $U \to V$
  in $\G$. On the other hand, if the token is now held by a
  \proponent{}, the corresponding sequent may have been anodynely
  modified. As observed above, we still may choose a cut-free proof
  for this sequent, as well as for the unmodified ones, and we add an
  edge $U \to V$ in $\G$.

  Now if the token is held by \proponents{} in $U$, we consider the
  first step of the proof of the involved sequent.  If it is a
  positive rule, we consider the corresponding move $U \to V$. In $V$,
  we still have cut-free proofs for all \proponents{}' sequents, and
  we add an edge $U \to V$ in $\G$.  If the proof starts with a
  negative rule, we add as an edge in $\G$ the move passing the token
  along the corresponding formula (in the new position, \proponents{}'
  sequents are unchanged, hence we still have cut-free proofs for
  them).
  
  This graph $\G$ determines a strategy on each of its vertices,
  since it accepts all opponent moves. And furthermore this strategy
  is total thanks to proofs, hence it is winning by
  \Lemma{lemma:total}.  
\qed\end{proof}

\begin{proof}[Proof of Theorem~\ref{thm:soundness}] Consider any
  formula $B$ with a bounded proof in \LLTN{}.  We
  have to choose a graph $U$ such that for all $V$, vertices in the
  position
  \begin{equation}
      \begin{hs}
    \player{l}{0,0} %
    \opponent{r}{2,0} %
    \subtreerotwith{l}{90}{$U$}{left} %
    \subtreerotwith{r}{-90}{$V$}{right} %
    \arrow{l}{r}{$B$} %
  \end{hs}\label{eq:poss}
\end{equation} 
have a winning strategy.

If the given proof of $B$ is cut-free, then Lemma~\ref{lemma:winning}
provides a winning strategy for $U = \Pl$, the single vertex (for
any $V$ and initial placement of the token).

  Otherwise the bounded proof has the shape
  \begin{mathpar}
      \inferrule{
    \inferrule{\pi_1}{A} \\
    \inferrule{\pi_2}{A^\bot, B} \\
  }{B.}
  \end{mathpar}
  We then choose $U$ to be
  \begin{center}
    \begin{hs}
      \player{l}{0,0} %
      \player{r}{2,0} %
      \arrow{l}{r}{$A$} %
    \end{hs}
  \end{center}
  so that for any $V$, the position~\eqref{eq:poss} becomes
  \begin{center}
    \begin{hs}
      \player{ll}{-2,0} %
      \player{l}{0,0} %
      \opponent{r}{2,0} %
      \subtreerotwith{r}{-90}{$V.$}{right} %
      \arrow{l}{r}{$B$} %
      \arrow{ll}{l}{$A$} %
    \end{hs}
  \end{center}  
  In this position, we have cut-free proofs ($\pi_1$ and $\pi_2$) for
  each proponent vertex, hence a winning strategy by
  \Lemma{lemma:winning} (for any placement of the token).
\end{proof}

However, we have:
\begin{proposition}\label{proposition:game:incomplete}
  \GLLTN{}-validity is complete neither w.r.t.\ \LL{}, nor w.r.t.\ \LLTN{}.
\end{proposition}
\begin{proof}
  Indeed, exactly as in the game for MALL~\cite{H3:mall:long}, $\bot
  \tens \bot$ is \GLLTN{}-valid but provable neither in \LL{} nor in \LLTN{}.
\qed\end{proof}
We expect that \emph{local} strategies~\cite{H3:mall:long} will remedy
incompleteness w.r.t.\ \LLTN{}.

\section{Admissibility of cut}\label{sec:cut}
We now define a notion of \GLLTN{}-validity for arbitrary sequents, which
extends that for formulae, and prove that the cut rule holds in our
model, i.e., if two sequents $\Gam, A$ and $A^\bot, \Del$ are \GLLTN{}-valid,
then so is $\Gam, \Del$.

\begin{definition}
  A sequent $\Gam = (A_1, \ldots, A_n)$ is \emph{\GLLTN{}-valid} iff
  there is a position $U$ with dangling edges $A_1, \ldots, A_n$, such
  that for all tuples of positions $U_1, \ldots, U_n$, the vertices in
  $U$ have a winning strategy $\strat$ on the position
  
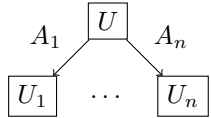
\begin{center}
    \begin{hs}
      \node[draw] (U) at (0,0) {$U$} ; %
      \node[draw] (U1) at (-1,-1) {$U_1$} ; %
      \node (dots) at (0,-1) {$\dots$} ; %
      \node[draw] (Un) at (1,-1) {$U_n$} ; %
      \arrowal{U}{U1}{$A_1$} %
      \arrowar{U}{Un}{$A_n$} %
    \end{hs}
  \end{center}
  We say that $U$ is its \emph{base} position.
\end{definition}

\begin{theorem}\label{thm:cut}
  The cut rule holds, i.e., if two sequents $\Gam, A$ and $A^\bot,
  \Del$ are \GLLTN{}-valid, then so is $\Gam, \Del$.
\end{theorem}
We first consider slightly generalised positions $U$, with three teams
instead of two, say, $P$, $P'$, and $O$. The teams $P$ and $P'$ are to
be thought of as a partition of \proponents{}.  For such a generalised
position $U$, call $(U, P)$ the non-generalised position with $P$ as
\proponents{} and $P' \cup O$ as opponents; and similarly for $(U,
P')$, $(U, P \cup P')$, and $(U, O)$.  We have:
\begin{lemma}
  For a generalised position $U$, given winning strategies $\strat$ on
  $(U, P)$ and $\strat'$ on $(U, P')$, the set of plays $\strat \cap
  \strat'$ is a winning strategy on $(U, P \cup P')$.
\end{lemma}
\begin{proof}
  It is a strategy because it is obviously prefix-closed, accepts all
  moves by $O$. Furthermore, it is total, hence winning by the
  previous lemma.
\qed\end{proof}

\begin{proof}[of Theorem~\ref{thm:cut}]
  Assume given two \GLLTN{}-valid sequents $\Gam, A$ and $A^\bot, \Del$, with
  associated base positions $U$ and $V$, and with $\Gam = (A_1, \ldots,
  A_n)$ and $\Del = (B_1, \ldots, B_m)$. Then, choose as a base for
  $\Gam, \Del$ the position
  \begin{center}
    \begin{hs}
      \node[draw] (U) at (0,0) {$U$} ; %
      \node[draw] (V) at (2,0) {$V$} ; %
      \arrow{U}{V}{$A$} %
    \end{hs}\ .    
  \end{center}
  Then, for any tuple $(U_1, \ldots, U_n, V_1, \ldots, V_m)$, 
  consider the corresponding position
  \begin{equation}
    \begin{hs}
      \node[draw] (U) at (0,0) {$U$} ; %
      \node[draw] (V) at (3,0) {$V$} ; %
      \arrow{U}{V}{$A$} %
      \node[draw] (U1) at (-2,0) {$U_1$} ; %
      \node[rotate=-30] (dots) at (-.8,-.5) {$\dots$} ; %
      \node[draw] (Un) at (0,-1.2) {$U_n$} ; %
      \arrow{U}{U1}{$A_1$} %
      \arrowr{U}{Un}{$A_n$} %
      \node[draw] (V1) at (3,-1.2) {$V_1$} ; %
      \node[rotate=30] (dots) at (3.8,-.5) {$\dots$} ; %
      \node[draw] (Vm) at (5,0) {$V_m$} ; %
      \arrowl{V}{V1}{$B_1$} ; %
      \arrow{V}{Vm}{$B_m$} ; %
    \end{hs}~.\label{eq:bigpos} 
  \end{equation}
  The \GLLTN{}-validity of $\Gam, A$ induces a winning strategy $S$
  for the vertices of $U$ in~\eqref{eq:bigpos}, and that of $A^\bot,
  \Del$ induces a winning strategy $S'$ for the vertices of $V$
  in~\eqref{eq:bigpos}.  This yields the winning strategy $S \cap S'$
  for $U \cup V$ on~\eqref{eq:bigpos}. Hence $\Gam, \Del$ is
  \GLLTN{}-valid.
\end{proof}

\bibliographystyle{plain}
\bibliography{b}

\end{document}